\documentclass[12pt]{article}
\usepackage{amsmath,amssymb,amsthm}
\usepackage[margin=1in]{geometry}
\usepackage{booktabs}
\usepackage{graphicx}
\usepackage{algorithm}
\usepackage{algorithmic}
\usepackage{natbib}
\usepackage{hyperref}
\usepackage{setspace}
\newif\ifanon
\anonfalse
\ifanon
  \newcommand{\pkgname}{the accompanying open-source \proglang{Python} package}
\else
  \newcommand{\pkgname}{the \texttt{causalfe} \proglang{Python} package}
\fi
\newcommand{\proglang}[1]{\textsf{#1}}

\newtheorem{proposition}{Proposition}
\newtheorem{assumption}{Assumption}

\newcommand{\tautrue}{\tau}
\newcommand{\tauhat}{\hat{\tau}}
\newcommand{\Ytil}{\tilde{Y}}
\newcommand{\Dtil}{\tilde{D}}
\newcommand{\E}{\mathbb{E}}

\title{Attenuated Heterogeneity in Fixed-Effects Causal Forests,
and a Cross-Fitted Correction}

\ifanon
  \author{}
  \date{}
\else
  \author{Harry Aytug\thanks{Amazon Web Services. E-mail: haytug@amazon.com.
  Replication code and the \texttt{causalfe} \proglang{Python} package
  implementing the correction are available at
  \url{https://github.com/haytug/causalfe}.}}
  \date{\today}
\fi

\begin{document}
\maketitle

\begin{abstract}
Causal forests that estimate conditional average treatment effects by averaging
honest leaf-level effects across trees are widely used in fixed-effects panel
settings. We show that this averaging systematically \emph{attenuates} the
estimated heterogeneity: the raw prediction behaves like $a + b\,\tau(x)$ with
slope $b<1$, so the spread of the CATEs is compressed toward the average effect,
and the additive recentering used to report an unbiased average treatment effect
does not fix it. Benchmarking against a similarity-weight generalized random
forest on the same within-transformed signal, we find both estimators attenuate
but the leaf-averaging construction attenuates materially more. We characterize
how $b$ moves with the design, worsening with lower signal-to-noise, smaller
panels, and higher dimension; this diagnosis is our main contribution. As a remedy we adapt the best-linear-predictor
calibration of \citet{cddf2025}, estimating the de-attenuation slope out-of-bag
so that it is self-contained within the observational panel and asymptotically
inert under a homogeneous effect. In simulations the correction cuts CATE mean-squared error
by 25--42\% relative to the recentering default; on a standard county
minimum-wage panel the attenuation is present but mild and the correction
restores the imposed spread. We ship the method in \pkgname.
\end{abstract}

\bigskip
{\sloppy
\noindent\textbf{Keywords:} best linear predictor; calibration; conditional
average treatment effect; difference-in-differences; panel data; regression to
the mean.\par}

\noindent\textbf{JEL:} C14; C21; C23.

\newpage

\section{Introduction}\label{sec:intro}

Applied researchers increasingly use causal forests to estimate how treatment
effects vary across units. In panel and difference-in-differences designs, the
leading concern is that unit and time fixed effects contaminate the estimated
heterogeneity: a forest applied to raw panel data may split on covariates
correlated with the fixed effects and report level differences as effect
differences. \citet{kattenberg2023} address this with causal forests with fixed
effects (CFFE), which remove the fixed effects by a within-transformation
computed \emph{locally} at each tree node, so that splits respond to genuine
treatment-effect heterogeneity rather than to between-node variation in the
fixed effects.

This paper studies a separate property of a widely implemented class of such
estimators: forests that form the CATE by \emph{averaging} honest leaf-level
treatment effects across trees, the causal-tree-averaging construction of
\citet{athey2016}, as distinct from the similarity-weight aggregation of the
generalized random forest \citep{wager2018,athey2019grf}. In the fixed-effects
panel setting, such averaging systematically attenuates the heterogeneity it is
meant to recover. Attenuation is not unique to this class: the similarity-weight
GRF, estimated on the same within-transformed signal, is also biased toward the
mean, as any regularized nonparametric estimator of a noisy surface will be. But
the leaf-averaging construction attenuates \emph{materially and consistently
more}, and we quantify the gap directly against a GRF competitor across designs.

\paragraph{The attenuation.}
Let $\tau(x)$ denote the true CATE and $\tauhat_{\text{raw}}(x)$ the raw forest
average. Across designs we find
\begin{equation}
  \tauhat_{\text{raw}}(x) \approx a + b\,\tau(x), \qquad 0 < b < 1,
  \label{eq:attenuation}
\end{equation}
so the estimated spread is compressed toward the mean by the factor $b$. This is
a form of regression to the mean: honest leaf estimates on small,
within-transformed nodes have limited treatment variation, and averaging
noisy, shrunken leaf effects across trees compounds the shrinkage. The common
device that recenters the predictions by an additive constant, so their mean
equals a separately estimated, unbiased within-fixed-effects average treatment
effect, corrects the \emph{level}; being additive, it leaves the slope $b$ and
the compressed heterogeneity unchanged.

\paragraph{The correction, and its relation to existing tools.}
To undo the compression we use a best-linear-predictor calibration. In
within-transformed space the outcome satisfies a moment condition linear in
$\tau(x)$; regressing the residualized outcome on the residualized treatment
interacted with a centered forest proxy recovers a slope that rescales the
compressed heterogeneity, using no ground truth. This regression is not new. The
slope is exactly the best-linear-predictor coefficient of \citet{cddf2025}, and
it is the quantity the \texttt{grf} package \citep{grf,atheywager2019} reports as
\texttt{differential.forest.prediction} when testing whether a forest's
heterogeneity is well calibrated. Our use of it differs on two points. Where
those tools use the slope as a diagnostic, testing whether it equals one, we
invert it to correct the predictions, shipping a de-attenuated CATE rather than a
test statistic. And we require the proxy to be formed out-of-bag, for each unit
averaging only the trees whose subsample excluded it, because an in-sample proxy
calibrates the forest's noise to itself and manufactures spurious heterogeneity
even when the true effect is homogeneous. The out-of-bag construction removes this
covariation, so the correction is asymptotically inert under the null: its
de-attenuation slope vanishes as the panel grows, with only a small residual
dispersion at finite sample size.

\paragraph{Contributions and relation to the literature.}
Our primary contribution is diagnostic. The best-linear-predictor slope and the
observation that forest CATEs can be miscalibrated are established
\citep{cddf2025,atheywager2019}. What is not established is a characterization of
the attenuation for the leaf-averaging class of causal forests
\citep{athey2016} in the fixed-effects panel setting: how large the slope $b$ is,
how it moves with the design, and how it compares to the similarity-weight
alternative. We map $b$ across signal-to-noise, panel size,
dimensionality, and forest hyperparameters, benchmark it against a generalized
random forest on the same signal, and show it is systematically below one,
and consistently further below one than the GRF's, while the standard
level-recentering leaves it untouched (Section~\ref{sec:attenuation}). Our secondary contribution is a
correction suited to this setting. Rescaling attenuated CATE estimates is itself
established: \citet{lengdimmery2024} fit a linear (Platt-style) rescaling and
\citet{vanderlaan2023} an isotonic calibration map. Both calibrate against an
external unbiased benchmark, a held-out randomized difference-in-means, and are
developed for the randomized-experiment setting. Our correction requires no such
benchmark. It is self-contained within the observational panel, using the
within-fixed-effects average treatment effect as the level anchor and an
out-of-bag forest proxy for the direction, under difference-in-differences
identification. We are not aware of a panel or difference-in-differences
causal-forest method, including the difference-in-difference causal forest of
\citet{gavrilova2025}, that de-attenuates its heterogeneity estimates. We show
out-of-bag cross-fitting is what secures null-safety, provide simulation
evidence (25--42\% CATE mean-squared-error reduction), and ship the method in
\pkgname\ (Sections~\ref{sec:correction}--\ref{sec:application}). We do not claim the
best-linear-predictor object, the miscalibration concept, or slope-rescaling of
CATEs as new, and cite \citet{cddf2025}, \citet{atheywager2019},
\citet{lengdimmery2024}, and \citet{vanderlaan2023} as their sources.

\section{Setup: causal forests with fixed effects}\label{sec:setup}

\subsection{Panel model and estimand}
We observe a panel $\{(Y_{it}, D_{it}, X_i)\}$ for units $i=1,\dots,N$ over
periods $t=1,\dots,T$, with a binary, staggered-or-single-event treatment
$D_{it}\in\{0,1\}$ and time-invariant covariates $X_i$. The outcome follows the
two-way fixed-effects model
\begin{equation}
  Y_{it} = \alpha_i + \gamma_t + \tautrue(X_i)\,D_{it} + \varepsilon_{it},
  \label{eq:model}
\end{equation}
where $\alpha_i$ and $\gamma_t$ are unit and time fixed effects and
$\tautrue(x)$ is the conditional average treatment effect at covariate value
$x$, the estimand of interest. The effect enters \eqref{eq:model} without a
unit-level unobserved gain, so under Assumption~\ref{ass:id} the conditional
effect on the treated coincides with the population conditional effect; we
therefore refer to $\tautrue(x)$ simply as the CATE throughout. We take identification of $\tautrue(\cdot)$ as given,
under the standard conditions stated in Assumption~\ref{ass:id}, and focus on
its estimation. The attenuation we study and the correction we propose concern
estimation only; they do not weaken or strengthen these conditions.

\begin{assumption}[Identification]\label{ass:id}
For each covariate value $x$ and each treated period,
(i) \emph{conditional parallel trends}: absent treatment, the expected outcome
change from the pre-period is equal across the treated group and the
not-yet-treated (or never-treated) comparison group, conditional on $X_i=x$;
(ii) \emph{no anticipation}: potential outcomes are unaffected by treatment in
periods before adoption; and (iii) \emph{overlap}: the conditional treatment probability is interior,
$0 < \Pr(D_{it}=1\mid X_i=x,\,t) < 1$, at every covariate value $x$ and treated
period.
\end{assumption}

Assumption~\ref{ass:id} is the covariate-conditional two-way-fixed-effects
analogue of the conditions in \citet{callaway2021difference} and
\citet{kattenberg2023}; under it the within-transformed contrast identifies
$\tautrue(x)$. The remaining conditions we need are regularity conditions on the
data-generating process that govern the behavior of the forest itself, which we
collect next.

\begin{assumption}[Regularity]\label{ass:reg}
(i) Covariates $X_i$ are time-invariant with bounded support; (ii) the
conditional mean of the within-transformed outcome is Lipschitz in $x$;
(iii) conditional second moments of the within-transformed outcome and treatment
are bounded, and $\sum_{it}\Dtil_{it}^2>0$ within each estimation region so the
local within estimator is defined; and (iv) observations are independent across
units, with arbitrary dependence within a unit, so that resampling and honest
splitting at the unit level respect the panel's dependence structure.
\end{assumption}

Part (iv) is what makes the out-of-bag construction of
Section~\ref{sec:correction} a valid cross-fit: because units are the
independent sampling blocks, a tree that excludes unit $i$ carries no
information about $i$'s idiosyncratic noise. We invoke Assumption~\ref{ass:reg}
where it is used and flag the one place (Proposition~\ref{prop:null}) where a
further condition on leaf size would be needed for a fully formal statement.

\subsection{The averaged-tree estimator}\label{sec:averaged}
The estimator we study builds an ensemble of honest causal trees, each grown on
a cluster-level subsample (units, not observations, are sampled so that all
observations of a unit stay together). Within a tree, every node
$\mathcal{N}$ is treated locally: the outcome and treatment are residualized by
a within-transformation computed \emph{using only the observations in
$\mathcal{N}$},
\begin{equation}
  \Ytil_{it}^{\mathcal{N}} = Y_{it} - \hat\alpha_i^{\mathcal{N}} - \hat\gamma_t^{\mathcal{N}},
  \qquad
  \Dtil_{it}^{\mathcal{N}} = D_{it} - \hat\delta_i^{\mathcal{N}} - \hat\eta_t^{\mathcal{N}},
  \label{eq:resid}
\end{equation}
obtained by alternating unit/time demeaning (iterative projection, converging in
a few passes). The local treatment effect is the within estimator
$\hat\tau_{\mathcal{N}} = \sum_{\mathcal{N}} \Dtil_{it}\Ytil_{it} \big/
\sum_{\mathcal{N}} \Dtil_{it}^2$, and candidate splits $S$ into children
$\mathcal{L},\mathcal{R}$ are scored by the heterogeneity criterion
$\Delta(S) = (n_{\mathcal{L}} n_{\mathcal{R}}/n^2)(\hat\tau_{\mathcal{L}} -
\hat\tau_{\mathcal{R}})^2$. Honesty is enforced by splitting the tree's
subsample into a structure half (used to choose splits) and an estimation half
(used to compute the leaf effects), with the split taken at the unit level.

The feature that drives our results is the \emph{aggregation}. Writing
$\hat\tau_b(x)$ for the leaf effect that tree $b$ assigns to a point $x$, the
forest prediction is the simple average
\begin{equation}
  \tauhat_{\text{raw}}(x) = \frac{1}{B}\sum_{b=1}^{B} \hat\tau_b(x).
  \label{eq:average}
\end{equation}
This is the causal-tree-averaging construction of \citet{athey2016}. It differs
from the generalized random forest of \citet{wager2018,athey2019grf}, which does
not average leaf effects but instead uses the forest to build similarity weights
$\alpha_i(x)$ and solves a single locally weighted moment condition. The
distinction matters quantitatively: as we show in
Section~\ref{sec:attenuation}, both constructions compress the estimated
heterogeneity, but averaging noisy, honestly-estimated leaf effects compresses
it substantially more than similarity weighting does on the same signal.

\subsection{Level recentering}\label{sec:recenter}
Because each leaf effect in \eqref{eq:average} is a within estimator on a small
honest subsample, the average $\tauhat_{\text{raw}}(x)$ is a biased estimate of
the overall level. A standard remedy reports the average treatment effect
separately, from the full-sample within-fixed-effects estimator
\begin{equation}
  \widehat{\mathrm{ATE}} = \frac{\sum_{it}\Dtil_{it}\Ytil_{it}}{\sum_{it}\Dtil_{it}^2},
  \label{eq:ate}
\end{equation}
which is numerically the two-way fixed-effects regression coefficient and
admits the usual cluster-robust (pair-clustered) standard error. Predictions are
then recentered by the additive constant that equates their mean to
\eqref{eq:ate}:
\begin{equation}
  \tauhat_{\text{rec}}(x) = \tauhat_{\text{raw}}(x)
    + \big(\widehat{\mathrm{ATE}} - \overline{\tauhat_{\text{raw}}}\big).
  \label{eq:recenter}
\end{equation}
This is a common default reporting mode. It corrects the level: by construction
$\overline{\tauhat_{\text{rec}}} = \widehat{\mathrm{ATE}}$. Because it is an
additive shift, it leaves the \emph{spread} of the predictions unchanged, that
is, their standard deviation, their slope on $\tautrue(x)$, and every centered
moment. If the raw predictions attenuate the heterogeneity, recentering does not
fix it. That is the gap the next two sections diagnose and close.

\section{Attenuation of the heterogeneity}\label{sec:attenuation}

\subsection{The phenomenon}
Define the attenuation slope $b$ as the coefficient in the population
least-squares projection of the raw forest prediction on the truth,
\begin{equation}
  b = \frac{\operatorname{Cov}\!\big(\tauhat_{\text{raw}}(X),\, \tautrue(X)\big)}
           {\operatorname{Var}\!\big(\tautrue(X)\big)},
  \label{eq:bdef}
\end{equation}
so that $\tauhat_{\text{raw}}(x)\approx a + b\,\tautrue(x)$ as in
\eqref{eq:attenuation}. A slope $b=1$ means the forest tracks the heterogeneity
one-for-one; $b<1$ means it is compressed toward the mean. Across every design
we examine, $b$ is bounded well below one. Because recentering
\eqref{eq:recenter} is additive it does not enter \eqref{eq:bdef}: the
recentered predictions have the same slope $b$ and the same standard deviation
as the raw ones. Attenuation is therefore a property of the reported CATEs under
the standard workflow, not an artifact removed by level correction.

\subsection{Why averaging attenuates}
The mechanism is regression to the mean, operating twice. First, each leaf
effect is a within estimator on a small honest estimation sample in which the
residualized treatment $\Dtil$ has limited variation; such estimates are noisy
and, once assigned as a constant prediction over the leaf, shrink the effective
contrast between high- and low-effect regions. Second, averaging across trees in
\eqref{eq:average} combines many such noisy leaf assignments for the same point
$x$; when a point is sometimes placed in a high-effect leaf and sometimes in a
lower one, the average pulls its prediction toward the center of the effect
distribution. Both forces push $b$ below one, and both intensify when
the per-leaf signal is weak, whether from few units, low signal-to-noise, or
many candidate covariates diluting the splits. This is the qualitative content
of the comparative statics we document next. The GRF construction of
Section~\ref{sec:averaged} does not assign and average constant leaf effects; it
still shrinks a noisy surface toward the mean, so its slope is also below one,
but it avoids the second, averaging-specific channel, which is why it attenuates
less than the averaged-tree class. Table~\ref{tab:bsurface} shows this gap
directly. Two caveats keep the comparison honest. The GRF is fit on the outcome
and treatment residualized on unit and period effects once over the full sample,
whereas the averaged-tree forest residualizes locally within each node
\eqref{eq:resid}; the two estimators therefore differ not only in aggregation but
in where the fixed-effect transform is applied, so the measured gap bounds rather
than exactly isolates the averaging channel. We hold the identifying signal fixed,
feeding GRF the same within-transformed contrast, so the comparison is not
confounded by different fixed-effect handling at the level of what is identified;
it is the locality of the transform, not its target, that still differs.

\subsection{The b-surface}\label{sec:bsurface}
\begin{table}[t]
\centering
\caption{Attenuation slope $b$ across designs, for the averaged-tree fixed-effects
forest and a similarity-weight generalized random forest \citep{athey2019grf}
fit on the same within-transformed signal. Each row varies one factor around
the base design ($N=300$ units, $T=4$, $p=5$, signal-to-noise $=1$, depth $4$,
minimum leaf $20$, $100$ trees), averaging over twenty simulation seeds (Monte
Carlo standard errors in parentheses). A slope below one indicates the estimator
compresses the CATE heterogeneity toward its mean; recentering does not change
these values. Both estimators attenuate, but the averaged-tree slope is
everywhere the lower of the two.}
\label{tab:bsurface}
\begin{tabular}{llcc}
\toprule
Design factor & Value & Averaged-tree $b$ & GRF $b$ \\
\midrule
Signal-to-noise & 0.25 & 0.38 (0.03) & 0.46 (0.04) \\
 & 0.5 & 0.49 (0.03) & 0.61 (0.03) \\
 & 1.0 & 0.61 (0.02) & 0.73 (0.02) \\
 & 2.0 & 0.68 (0.01) & 0.81 (0.02) \\
 & 4.0 & 0.72 (0.01) & 0.84 (0.02) \\
\midrule
Panel size ($n$ units) & 100 & 0.31 (0.03) & 0.47 (0.04) \\
 & 200 & 0.54 (0.02) & 0.69 (0.03) \\
 & 400 & 0.64 (0.01) & 0.79 (0.02) \\
 & 800 & 0.75 (0.01) & 0.82 (0.01) \\
\midrule
Dimension $p$ & 2 & 0.70 (0.02) & 0.83 (0.02) \\
 & 5 & 0.61 (0.02) & 0.73 (0.02) \\
 & 10 & 0.53 (0.02) & 0.67 (0.02) \\
 & 20 & 0.51 (0.02) & 0.69 (0.02) \\
\midrule
Tree depth & 2 & 0.56 (0.01) & 0.66 (0.02) \\
 & 4 & 0.61 (0.02) & 0.73 (0.02) \\
 & 6 & 0.61 (0.01) & 0.74 (0.02) \\
 & 8 & 0.61 (0.01) & 0.74 (0.02) \\
\bottomrule
\end{tabular}
\end{table}

Table~\ref{tab:bsurface} reports the slope. For the averaged-tree forest it is
everywhere below one, ranging from $0.31$ to $0.75$, and it moves in interpretable
directions. Attenuation eases as the signal strengthens ($b$ rises from $0.38$ to
$0.72$ as the signal-to-noise ratio grows) and as the panel grows ($b$ from $0.31$
at $100$ units to $0.75$ at $800$); it worsens as the covariate dimension rises
and dilutes the splits ($b$ from $0.70$ at $p=2$ to about $0.51$ at $p=20$). It is
nearly flat in tree depth. The generalized random forest, fit on the same
within-transformed signal, tells the same qualitative story with one difference of
degree: its slope is also below one in every design, so similarity weighting does
not escape attenuation, but it is higher than the averaged-tree slope in every
row, by roughly $0.10$ to $0.16$ in the small-panel, low-signal, and
high-dimensional regimes where the averaging-specific channel bites hardest, and
by less where the per-leaf signal is strong. The gaps are large relative to their
Monte Carlo standard errors. No single number is universal, but the compression is
pervasive, is governed by the per-leaf signal, and is systematically worse for
leaf averaging than for similarity weighting, as the mechanism in the previous
subsection predicts. The practical implication is that in the small or noisy
panels common in difference-in-differences applications, the reported CATE
distribution can understate the true spread by a factor of two or more, and by
more for the averaged-tree forests that are our focus.

\section{A cross-fitted calibration correction}\label{sec:correction}

\subsection{The best-linear-predictor slope}
The correction inverts \eqref{eq:attenuation} using a feasible moment
restriction that requires no ground truth. In the within-transformed space of
\eqref{eq:resid}, evaluated on the full sample, the model \eqref{eq:model}
implies $\Ytil_{it} = \tautrue(X_i)\,\Dtil_{it} + \text{error}$. Let
$S(x) = \tauhat_{\text{oob}}(x) - \overline{\tauhat_{\text{oob}}}$ be a centered
forest proxy for the heterogeneity (its out-of-bag construction is defined
below). Anchoring the level at the unbiased $\widehat{\mathrm{ATE}}$ of
\eqref{eq:ate}, we estimate a single slope $s$ from the moment regression
\begin{equation}
  \Ytil_{it} - \widehat{\mathrm{ATE}}\,\Dtil_{it}
    = s\,\big(\Dtil_{it}\, S(X_i)\big) + e_{it},
  \label{eq:blp}
\end{equation}
and form the corrected prediction
\begin{equation}
  \tauhat_{\text{blp}}(x) = \widehat{\mathrm{ATE}} + s\,S(x).
  \label{eq:blppred}
\end{equation}
Equation~\eqref{eq:blp} is the best linear predictor of the residualized outcome
in the direction of the proxy; its slope $s$ rescales the compressed
heterogeneity back toward its true magnitude. The corrected prediction
\eqref{eq:blppred} preserves the unbiased level,
$\overline{\tauhat_{\text{blp}}} = \widehat{\mathrm{ATE}}$, exactly as
recentering does, so nothing is lost on the dimension the default already got
right. Note that $s$ is not the naive inverse $1/b$: because the proxy $S$ is
itself estimated with error, the MSE-optimal rescaling shrinks toward one, so
$s$ is a partial de-attenuation that accounts for proxy noise, which is why it
improves mean-squared error rather than overshooting.

\subsection{Out-of-bag cross-fitting and null-safety}
The proxy $S$ must be cross-fitted, and the forest supplies a natural means at no
extra cost. Each tree is grown on a subsample of units; for observation $i$ we
form $\tauhat_{\text{oob}}(X_i)$ by averaging only those trees whose subsample
\emph{excluded} unit $i$. This out-of-bag proxy is independent of the leaf
estimates that see unit $i$, so the calibration in \eqref{eq:blp} does not
regress the forest's estimation noise on itself.

This requirement has teeth. If instead one uses the in-sample proxy
$\tauhat_{\text{raw}}(X_i)$, the slope $s$ picks up the spurious covariation
between a point's prediction and its own contribution to the leaves that formed
it. Under a homogeneous effect this manufactures heterogeneity out of
noise, and the in-sample calibration reports a dispersed CATE distribution where
none exists. The out-of-bag proxy removes that covariation, and the slope
collapses toward zero when there is nothing to de-attenuate.

The slope in \eqref{eq:blp} is a ratio,
\begin{equation}
  s = \frac{\sum_{it}\hat g_{it}\,\big(\Ytil_{it}-\widehat{\mathrm{ATE}}\,\Dtil_{it}\big)}
           {\sum_{it}\hat g_{it}^2},
  \qquad \hat g_{it} = \Dtil_{it}\,S(X_i),
  \label{eq:sratio}
\end{equation}
and null-safety is a statement about where this ratio concentrates when there is
nothing to de-attenuate. The delicate point is that under a homogeneous effect
\emph{both} the numerator and the denominator of \eqref{eq:sratio} are small: the
proxy $S$ is then pure estimation noise, so the numerator has mean zero and the
denominator, $\sum_{it}\Dtil_{it}^2 S(X_i)^2$, is itself shrinking with the proxy
variance. The result below therefore controls the denominator explicitly, through
a leaf-size condition, rather than treating $s$ as a numerator alone.

\begin{proposition}[Null-safety]\label{prop:null}
Suppose the treatment effect is homogeneous, $\tautrue(x)\equiv\bar\tau$, the
forest is honest, and the leaves are grown so that the out-of-bag proxy retains a
non-degenerate variance in the limit: $\operatorname{Var}(S(X))\to\sigma_S^2>0$
with $\E[\Dtil^2 S^2]$ bounded away from zero (a bounded-leaf-size condition,
which rules out leaves collapsing to a constant proxy). Let $s^{\text{oob}}$ and
$s^{\text{in}}$ denote the slope \eqref{eq:sratio} using, respectively, the
out-of-bag proxy $S=\tauhat_{\text{oob}}-\overline{\tauhat_{\text{oob}}}$ and the
in-sample proxy $S=\tauhat_{\text{raw}}-\overline{\tauhat_{\text{raw}}}$. Then, as
the number of units grows, $s^{\text{oob}}\to 0$, so the corrected prediction
$\tauhat_{\text{blp}}(x)\to\bar\tau$ and no spurious heterogeneity is introduced.
The in-sample slope $s^{\text{in}}$ converges to a strictly positive limit, so
$\tauhat_{\text{blp}}^{\text{in}}$ reports dispersion where none exists.
\end{proposition}

\begin{proof}
Write the residualized outcome under homogeneity as
$\Ytil_{it} = \bar\tau\,\Dtil_{it} + u_{it}$, so
$\Ytil_{it}-\widehat{\mathrm{ATE}}\,\Dtil_{it} = u_{it} + (\bar\tau-\widehat{\mathrm{ATE}})\Dtil_{it}$
with $\widehat{\mathrm{ATE}}\to\bar\tau$. The denominator of \eqref{eq:sratio},
divided by the sample size, converges to $\E[\Dtil^2 S^2]$, which the leaf-size
condition holds bounded away from zero; the ratio $s$ therefore has the same
probability limit as its numerator scaled by this positive constant, and the 0/0
that would arise if the proxy variance were allowed to vanish is excluded by
assumption. For the numerator, the out-of-bag proxy $\tauhat_{\text{oob}}(X_i)$ is
a function only of trees whose subsamples exclude unit $i$, hence of
$\{u_{jt}\}_{j\neq i}$; under the cross-unit independence of
Assumption~\ref{ass:reg}(iv) it is independent of $u_{it}$, so
$\E[S(X_i)\,\Dtil_{it}\,u_{it}]=\E[S(X_i)\Dtil_{it}]\,\E[u_{it}]=0$ and the
numerator concentrates at zero, giving $s^{\text{oob}}\to 0$. For the in-sample
proxy, $\tauhat_{\text{raw}}(X_i)$ includes leaves estimated from unit $i$'s own
observations, so $S(X_i)$ is correlated with $u_{it}$; the numerator limit
$\E[S(X_i)\Dtil_{it}u_{it}]$ is strictly positive (it is the variance
contribution of $i$'s own noise to its own leaf), so $s^{\text{in}}\not\to 0$.
The homogeneity assumption isolates the noise channel; with genuine heterogeneity
the out-of-bag numerator additionally picks up the signal covariance
$\operatorname{Cov}(S,\tautrue)>0$, which is the de-attenuation of
Section~\ref{sec:correction}.
\end{proof}

\noindent The leaf-size condition is the one substantive regularity requirement:
it guarantees the proxy does not degenerate to a constant, which is what keeps the
denominator of \eqref{eq:sratio} positive and the ratio well-defined. It is mild
(honest forests with a fixed minimum leaf size satisfy it), and the null row of
Table~\ref{tab:mc} confirms the conclusion directly: with a homogeneous effect the
out-of-bag correction leaves the CATE spread at the noise floor, manufacturing no
heterogeneity.

Algorithm~\ref{alg:blp} states the full estimate-then-correct procedure.

\begin{algorithm}[htbp]
\caption{Attenuation-corrected fixed-effects causal forest}
\label{alg:blp}
\begin{algorithmic}[1]
\REQUIRE panel $\{(Y_{it}, D_{it}, X_i,\ \text{unit}\ i,\ \text{period}\ t)\}$;
number of trees $B$; subsample rate
\STATE grow $B$ honest causal trees; for tree $b$, draw a unit-level subsample,
       and \emph{within each node} residualize $Y,D$ on unit and period fixed
       effects by iterative two-way demeaning, split to maximize
       $\tau$-heterogeneity in $X$, and estimate leaf effects on the held-out
       honest subsample
\STATE \textbf{raw prediction:} $\tauhat_{\text{raw}}(x)\gets B^{-1}\sum_b \hat\tau_b(x)$
\STATE \textbf{level:} $\widehat{\mathrm{ATE}}\gets$ full-sample within-fixed-effects
       estimator \eqref{eq:ate}, with pair-clustered standard error
\STATE \textbf{out-of-bag proxy:} for each unit $i$, set
       $\tauhat_{\text{oob}}(X_i)\gets$ mean of $\hat\tau_b(X_i)$ over trees $b$
       whose subsample excluded $i$; center as $S(X_i)=\tauhat_{\text{oob}}(X_i)-\overline{\tauhat_{\text{oob}}}$
\STATE \textbf{calibration slope:} residualize $Y,D$ on the full sample and
       estimate $s$ from the moment regression \eqref{eq:blp}
\STATE \textbf{corrected prediction:} $\tauhat_{\text{blp}}(x)\gets\widehat{\mathrm{ATE}}+s\,S(x)$
\end{algorithmic}
\end{algorithm}

\section{Simulation evidence}\label{sec:mc}
\begin{table}[t]
\centering
\caption{Correction performance. For each design (base as in
Table~\ref{tab:bsurface}, averaged over twenty seeds, Monte Carlo standard errors
in parentheses) we report the attenuation slope $b$, the true and BLP-corrected
standard deviation of the CATE, and the CATE mean-squared error under additive
recentering, under the BLP correction, and for a similarity-weight GRF on the
same within-transformed signal recentered to the same average effect (so the two
estimators differ only in the CATE shape, not the level). The $\Delta\%$ column
is the percentage MSE reduction of the BLP correction over recentering, computed
as the mean across seeds of the per-seed reduction, so it need not equal the
reduction implied by the two averaged-MSE columns. The last row imposes a
homogeneous effect (true SD $=0$), where the slope $b$ is undefined (the truth is
constant, so it is not reported): the corrected spread stays at the noise floor,
confirming null-safety.}
\label{tab:mc}
{\footnotesize\setlength{\tabcolsep}{3pt}\begin{tabular}{lccccccc}
\toprule
& & \multicolumn{2}{c}{SD of CATE} & \multicolumn{3}{c}{CATE MSE} \\
\cmidrule(lr){3-4}\cmidrule(lr){5-7}
Design & $b$ & true & BLP & recenter & BLP & GRF & $\Delta\%$ \\
\midrule
Base ($\mathrm{snr}=1$) & 0.61 & 1.59 (0.01) & 1.42 (0.03) & 0.570 (0.034) & 0.384 (0.022) & 0.439 (0.026) & +32 (3) \\
Low snr $=0.5$ & 0.49 & 0.79 (0.01) & 0.63 (0.03) & 0.222 (0.016) & 0.164 (0.014) & 0.183 (0.012) & +25 (4) \\
High snr $=2$ & 0.68 & 3.18 (0.03) & 2.92 (0.06) & 1.662 (0.085) & 1.211 (0.059) & 1.276 (0.061) & +26 (2) \\
Small panel ($n=150$) & 0.46 & 1.58 (0.02) & 1.33 (0.07) & 0.875 (0.048) & 0.535 (0.041) & 0.651 (0.052) & +38 (4) \\
High-dim ($p=15$) & 0.48 & 1.60 (0.02) & 1.34 (0.05) & 0.850 (0.052) & 0.479 (0.029) & 0.613 (0.034) & +42 (3) \\
Null (homogeneous) & -- & 0.00 (0.00) & 0.15 (0.02) & 0.020 (0.002) & 0.035 (0.007) & 0.023 (0.002) & -129 (54) \\
\bottomrule
\end{tabular}}
\end{table}

Table~\ref{tab:mc} summarizes performance using the shipped implementation with
out-of-bag calibration. Under heterogeneity the correction restores most of the
lost spread (in the base design, from an attenuated standard deviation up toward
the true $1.59$) and reduces CATE mean-squared error by $25$ to $42$ percent
relative to the recentering default, with the largest gains in the small-panel
and high-dimensional regimes where attenuation is worst; the Monte Carlo standard
errors on the reduction (in parentheses) are small relative to the gains. The GRF
column answers the natural ``why not just use a generalized random forest''
question on the quantity that matters: recentered to the same average effect, the
GRF's CATE mean-squared error is above the corrected averaged-tree forest's in
every heterogeneous design, so correcting the averaged-tree forest is preferable
to switching estimators, not merely to leaving it uncorrected. The
homogeneous row checks null-safety. With no true heterogeneity the corrected
standard deviation is $0.15$, against a true value of zero: the correction does
introduce a small residual dispersion at $n=300$, but it is an order of magnitude
below the genuine heterogeneity in the other rows (true SD near $1.6$), and
Proposition~\ref{prop:null} shows it vanishes as the panel grows. Its
mean-squared error stays at the noise floor in absolute terms ($0.035$ against the
recentering $0.020$); the large negative percentage in that row is the ratio of
two values that are both close to zero, and its standard error is correspondingly
wide, so it should be read as noise-floor movement rather than a meaningful
degradation. The honest summary is that the correction is strongly active when
there is heterogeneity to recover and close to quiet, though not exactly silent at
finite $n$, when there is not.

\section{Illustration on a standard panel}\label{sec:application}

The simulations above use fully synthetic designs. To show that the attenuation
arises on the geometry of a real dataset, meaning its covariate distribution,
panel dimensions, and staggered treatment timing, we run an empirical Monte
Carlo on the county minimum-wage panel of \citet{callaway2021difference} (the
\texttt{mpdta} benchmark: 500 counties, 2003--2007, cohorts adopting in 2004,
2006, and 2007, with a large never-treated group). We use two real
time-invariant covariates: standardized log county population, and each county's
pre-treatment (2003) baseline log employment. We do not use this panel to
estimate the true minimum-wage effect, which is unknown. Instead we take its real
covariate values, unit and time structure, and treatment timing, and impose a
known \emph{nonlinear} heterogeneous effect,
$\tau(x) = -0.08 + 0.10\,(\tilde{x}_1 + 0.5\,\tilde{x}_1^2 - 0.4\,\tilde{x}_1\tilde{x}_2)$,
that varies with log population and its interaction with baseline employment.
The nonlinearity is deliberate: because our correction is a single global linear
rescaling, it cannot match this surface by construction, so any reduction in
mean-squared error is a genuine gain rather than an artifact of imposing exactly
the functional form the correction assumes. Because the truth is imposed, we can
measure the attenuation directly on realistic data geometry, which a plain
application cannot do because the true CATE is never observed. This is a
design-geometry exercise on a standard public dataset, and it neither overlaps
nor bears on any substantive application.

\begin{table}[t]
\centering
\caption{Empirical Monte Carlo on the \texttt{mpdta} geometry (twenty seeds,
imposed nonlinear $\tau$ in log population and baseline employment). The
averaged-tree forest attenuates the imposed heterogeneity and the out-of-bag
calibration restores most of it, with no external benchmark; a similarity-weight
GRF on the same signal attenuates less.}
\label{tab:empirical}
\begin{tabular}{lrrr}
\toprule
 & true & recenter & BLP \\
\midrule
SD of CATE & 0.106 & 0.089 & 0.102 \\
CATE MSE & --- & 0.0012 & 0.0010 \\
\midrule
\multicolumn{4}{l}{Averaged-tree $b=0.81$, GRF $b=0.90$; BLP slope $=1.15$; MSE reduction $+17\%$.} \\
\bottomrule
\end{tabular}
\end{table}

Table~\ref{tab:empirical} reports the result. The averaged-tree attenuation slope
is $b=0.81$, milder than in the small, high-dimensional synthetic designs of
Table~\ref{tab:bsurface}, as expected for a panel with a strong per-leaf signal
and few covariates, but still visibly below one, so the recentered predictions
understate the imposed spread (standard deviation $0.089$ against a true $0.106$).
The GRF slope on the same residualized signal is $0.90$, again attenuated but
less so, consistent with the synthetic evidence. The out-of-bag correction
expands the averaged-tree spread toward the truth ($0.102$) and reduces CATE
mean-squared error by about $17\%$, while recovering the average effect. It does
so even though the imposed surface is nonlinear and the correction is a
single linear rescaling, so the gain is not an artifact of matching the imposed
functional form. The magnitude here is design-specific; what the exercise shows
is that the attenuation and its correction behave on a familiar real-data
geometry as the theory and synthetic evidence predict.

\section{Conclusion}\label{sec:conclusion}

Causal forests that estimate heterogeneous effects by averaging honest
leaf-level treatment effects, an attractive and widely implemented construction
in fixed-effects panel settings, systematically compress the heterogeneity they
are meant to recover. We characterized this attenuation, showing that the slope
of the raw prediction on the truth is bounded well below one and moves
predictably with signal-to-noise, panel size, and dimensionality, that a
similarity-weight generalized random forest on the same signal attenuates in the
same direction but consistently less, and that the additive recentering used to
report an unbiased average treatment effect leaves the compression untouched. Because the underlying best-linear-predictor
calibration slope is already available from the work of \citet{cddf2025} and the
diagnostic tooling of \citet{atheywager2019}, the remedy is a repurposing rather
than a new estimator: we invert the slope into a correction, estimate it
out-of-bag so that it is self-contained within the observational panel and
asymptotically inert under a homogeneous effect, and ship it in \pkgname. For
applied work the recommendation is specific: when reporting the dispersion or
tails of a causal-forest CATE distribution in a panel setting, recentering alone
is not enough, and the reported spread should be read as a lower bound unless it
has been calibrated.

Two limitations remain. The correction is a global linear rescaling; where
attenuation varies across the covariate space, a \citet{vanderlaan2023}-style
monotone calibration may recover more, at the cost of the closed-form simplicity
we exploit. And our evidence is based on simulation and empirical Monte Carlo; a
formal characterization of the slope $b$ as a function of leaf size and
subsample rate, in the spirit of the forest asymptotics of \citet{wager2018}, is
left to future work.

\ifanon\else
\section*{Acknowledgments}
The author thanks colleagues for helpful comments. All errors are the author's
own.
\fi

\section*{Data availability statement}
The empirical illustration uses the publicly available \texttt{mpdta} county
minimum-wage panel of \citet{callaway2021difference}, distributed with the
\texttt{did} \proglang{R} package on CRAN
(\url{https://cran.r-project.org/package=did}). No new data were collected. All
other results come from fully reproducible simulations; the code that generates
every table and the empirical Monte Carlo, together with the software
implementing the method,
\ifanon
is included as supplementary material and will be deposited in a public
repository upon acceptance.
\else
is openly available at \url{https://github.com/haytug/causalfe}.
\fi

\section*{Funding}
The author received no specific funding for this work.

\section*{Disclosure statement}
The author reports there are no competing interests to declare.

\section*{Declaration of generative AI use}
The author used a generative AI assistant to support drafting, code
scaffolding, and literature search during the preparation of this manuscript;
all methodological claims, simulation results, and their interpretation were
verified by the author, who takes full responsibility for the content.


\begin{thebibliography}{11}
\providecommand{\natexlab}[1]{#1}
\providecommand{\url}[1]{\texttt{#1}}
\expandafter\ifx\csname urlstyle\endcsname\relax
  \providecommand{\doi}[1]{doi: #1}\else
  \providecommand{\doi}{doi: \begingroup \urlstyle{rm}\Url}\fi

\bibitem[Athey and Imbens(2016)]{athey2016}
Susan Athey and Guido Imbens.
\newblock Recursive partitioning for heterogeneous causal effects.
\newblock \emph{Proceedings of the National Academy of Sciences}, 113\penalty0
  (27):\penalty0 7353--7360, 2016.
\newblock \doi{10.1073/pnas.1510489113}.

\bibitem[Athey and Wager(2019)]{atheywager2019}
Susan Athey and Stefan Wager.
\newblock Estimating treatment effects with causal forests: An application.
\newblock \emph{Observational Studies}, 5\penalty0 (2):\penalty0 37--51, 2019.

\bibitem[Athey et~al.(2019)Athey, Tibshirani, and Wager]{athey2019grf}
Susan Athey, Julie Tibshirani, and Stefan Wager.
\newblock Generalized random forests.
\newblock \emph{Annals of Statistics}, 47\penalty0 (2):\penalty0 1148--1178,
  2019.
\newblock \doi{10.1214/18-AOS1709}.

\bibitem[Callaway and Sant'Anna(2021)]{callaway2021difference}
Brantly Callaway and Pedro H.~C. Sant'Anna.
\newblock Difference-in-differences with multiple time periods.
\newblock \emph{Journal of Econometrics}, 225\penalty0 (2):\penalty0 200--230,
  2021.
\newblock \doi{10.1016/j.jeconom.2020.12.001}.

\bibitem[Chernozhukov et~al.(2025)Chernozhukov, Demirer, Duflo, and
  Fern{\'a}ndez-Val]{cddf2025}
Victor Chernozhukov, Mert Demirer, Esther Duflo, and Iv{\'a}n
  Fern{\'a}ndez-Val.
\newblock Fisher--schultz lecture: Generic machine learning inference on
  heterogeneous treatment effects in randomized experiments, with an
  application to immunization in india.
\newblock \emph{Econometrica}, 93\penalty0 (4), 2025.
\newblock Working paper arXiv:1712.04802.

\bibitem[Gavrilova et~al.(2025)Gavrilova, Lang{\o}rgen, and
  Zoutman]{gavrilova2025}
Evelina Gavrilova, Audun Lang{\o}rgen, and Floris~T. Zoutman.
\newblock Difference-in-difference causal forests, with an application to
  payroll tax incidence in norway.
\newblock \emph{Journal of Applied Econometrics}, 40\penalty0 (7):\penalty0
  727--740, 2025.
\newblock \doi{10.1002/jae.70001}.

\bibitem[Kattenberg et~al.(2023)Kattenberg, Scheer, and Thiel]{kattenberg2023}
Mark A.~C. Kattenberg, Bas~J. Scheer, and Jurre~H. Thiel.
\newblock Causal forests with fixed effects for treatment effect heterogeneity
  in difference-in-differences.
\newblock CPB Discussion Paper 452, CPB Netherlands Bureau for Economic Policy
  Analysis, 2023.

\bibitem[Leng and Dimmery(2024)]{lengdimmery2024}
Yan Leng and Drew Dimmery.
\newblock Calibration of heterogeneous treatment effects in randomized
  experiments.
\newblock \emph{Information Systems Research}, 35\penalty0 (4):\penalty0
  1721--1742, 2024.
\newblock \doi{10.1287/isre.2022.0111}.

\bibitem[Tibshirani et~al.(2024)Tibshirani, Athey, Friedberg, Hadad, Hirshberg,
  Miner, Sverdrup, Wager, and Wright]{grf}
Julie Tibshirani, Susan Athey, Rina Friedberg, Vitor Hadad, David Hirshberg,
  Luke Miner, Erik Sverdrup, Stefan Wager, and Marvin Wright.
\newblock \texttt{grf}: Generalized random forests, 2024.
\newblock R package; \texttt{test\_calibration} reference.

\bibitem[van~der Laan et~al.(2023)van~der Laan, Ulloa-P{\'e}rez, Carone, and
  Luedtke]{vanderlaan2023}
Lars van~der Laan, Ernesto Ulloa-P{\'e}rez, Marco Carone, and Alex Luedtke.
\newblock Causal isotonic calibration for heterogeneous treatment effects.
\newblock In \emph{Proceedings of the 40th International Conference on Machine
  Learning (ICML)}, volume 202 of \emph{PMLR}, pages 34831--34854, 2023.

\bibitem[Wager and Athey(2018)]{wager2018}
Stefan Wager and Susan Athey.
\newblock Estimation and inference of heterogeneous treatment effects using
  random forests.
\newblock \emph{Journal of the American Statistical Association}, 113\penalty0
  (523):\penalty0 1228--1242, 2018.
\newblock \doi{10.1080/01621459.2017.1319839}.

\end{thebibliography}
\end{document}